\documentclass[11pt]{article}
\usepackage{chaithanya}
\title{Lee-Yang theorem for fermions}
\author{Chaithanya Rayudu\thanks{University of Cambridge, sscr2@cam.ac.uk}
\and
Takahiro Misawa\thanks{University of Tokyo, \{tmisawa, juntakahashi\}@issp.u-tokyo.ac.jp}
\and
Andrew Zhao\thanks{Sandia National Laboratories, azhao@sandia.gov}
\and
Jun Takahashi\footnotemark[2]
}
\date{}

\begin{document}

\maketitle

\begin{abstract}
    Lee-Yang theorems are a powerful tool for studying many-body systems, with applications ranging from analyzing phase transitions to proving the efficiency of certain classical and quantum algorithms. In this work, we prove a Lee-Yang zero-freeness theorem for the partition function of a broad class of interacting fermion models, implying the existence of a provably efficient quantum algorithm for estimating their ground-state energies. This class includes several well-known models such as the attractive Hubbard model, repulsive Hubbard model on bipartite graphs, and the interacting Hofstadter model. Our results also rigorously establish the nonexistence of phase transitions in these models in the presence of a nonzero local external field.
\end{abstract}

\section{\label{sec:intro}Introduction}
The celebrated Lee-Yang theorem \cite{lee1952} states that the complex zeros of the partition function of the ferromagnetic Ising model, regardless of its spatial dimension or interaction strengths, lie on the imaginary axis in the complex plane of the magnetic field.
Physically, this zero-freeness of the partition function implies that the free energy is analytic, which in turn implies the nonexistence of phase transitions. This has provided a powerful mathematical tool to analytically understand phase transitions, or lack thereof, in physical systems.
Consequently, it has spurred on a plethora of works that probe classical and quantum phase transitions via the study of complex zeros of the partition function 
\cite{fisher1965nature,bena2005statistical,heyl2018dynamical}. 
More than 70 years later today, the Lee-Yang theorem continues to inspire new directions in physics research, ranging from high-energy theory \cite{mukherjee2021universality, wada2025locating} to condensed-matter experiments \cite{gao2024experimental}. 

A notable development in theoretical computer science in recent years was the realization that efficiency guarantees for classical algorithms can be obtained from the Lee-Yang theorem and its variants \cite{barvinok2016combinatorics,patel2017deterministic}. 
Physically, this can be understood as the lack of certain phase transitions allowing for techniques such as series expansions to converge efficiently, or for Markov chains to mix quickly  \cite{chen2024spectral}.  More specific to quantum systems, Asano \cite{asano1970} extended the Lee-Yang theorem to the ferromagnetic quantum Heisenberg model, and this was
subsequently fully generalized to a broader class of two-body interacting quantum spin systems by Suzuki and Fisher \cite{suzuki1971}.
These extensions have led to classical algorithms for a subclass of such Hamiltonians \cite{harrow2020classical}, and more recently, to provably efficient quantum algorithms for the entire class \cite{rayudu2026spectral,bravyi2026efficient}. 

The remarkable power of Lee-Yang-type theorems is that they allow us to have rigorous mathematical control of strongly interacting systems far away from perturbative regimes, a rare property to come across in the study of quantum many-body systems. 
This naturally leads to the question of how far this zero-freeness property extends to. 
For instance, although various studies have computed partition function zeros for specific fermionic systems \cite{PhysRevB.53.7704,wakayama2019lee,shastry2025partition}, a rigorous zero-freeness theorem for fermions is surprisingly unavailable. In fact, such a theorem is unknown even for the simplest possible case of noninteracting free-fermion models. This gap in understanding drives our motivation for the present work: to establish a Lee-Yang theorem for interacting fermions and classify such systems that exhibit this zero-freeness property.

\subsection{Main results}

We consider the following broad class of interacting fermionic Hamiltonians and prove a Lee-Yang zero-freeness theorem for all the Hamiltonians in this class. For $v \in \C^N$, let $a^\dagger(v) := \sum_i v_i\, a^\dagger_i$ be the creation operator of the mode $v$, and let $n(v) := a^\dagger(v)\, a(v)$. For an orthonormal pair of modes $v, w$, define
\begin{equation}\label{eq:Q_operators}
    q^{\pm}_{v,w} := \frac{1}{2}\Big(\I - n(v) - n(w)\Big) \pm \frac{1}{2}\Big(a^\dagger(v)\, a^\dagger(w) - a(w)\, a(v)\Big),
\end{equation}
and let $\mathcal{Q}$ denote the set of all operators $q^{\pm}_{v,w}$, over both signs and all orthonormal pairs $(v,w)$. Orthogonality is required only between $v$ and $w$ within a single operator $q^{\pm}_{v,w}$. The pairs defining different elements of $\mathcal{Q}$ are arbitrary and need not be orthogonal to one another. Note that $(q^\pm_{v,w})^\dagger = q^\mp_{v,w}$, so $\mathcal Q$ is closed under taking adjoints, and $(q^\pm_{v,w})^2=0$. 
We further discuss how $\mathcal Q$ arises in \cref{sec:discussion}. 

\begin{definition}\label{def:lee-yang}
The class of \emph{fermionic Lee-Yang Hamiltonians} is defined as
\begin{subequations}
    \label{eq:free_fermion_and_interactions}
    \begin{gather}
        H = H_0 + V \\[6pt]
        \text{where} \quad H_0 := \sum_{i,j} A_{ij}\, a^{\dagger}_i a_j \,+\, \frac{1}{2}\sum_{i,j}\,\left(\,B_{ij}\, a^{\dagger}_i a^{\dagger}_j + B_{ij}^{*}\, a_j a_i\,\right),\label{eq:free_fermion}\\
        V := - \sum_{k}\, U_k \, \big(r_k\, s_k + s^{\dagger}_{k}\, r^{\dagger}_{k}\big) \label{eq:interactions}
    \end{gather}
\end{subequations}
with a positive-semidefinite (PSD) hopping matrix $A = A^{\dagger} \succeq 0$, an arbitrary antisymmetric pairing matrix $B = -B^T$, finitely many terms with arbitrary nonnegative couplings $U_k \geq 0$, and arbitrary $r_k, s_k \in \mathcal{Q}$.
\end{definition}

We can check that, by picking $(r_k, s_k) = (q^+_{v,w}, q^-_{v,w})$ and $(q^-_{v,w}, q^+_{v,w})$ with equal couplings, the two corresponding terms in \cref{eq:interactions} add up to
\begin{equation}\label{eq:density_density_from_Q}
    -2U_k\left(q^+_{v,w}\, q^-_{v,w} + q^-_{v,w}\, q^+_{v,w}\right) \,=\, -4U_k\left(n(v) - \frac{1}{2}\right)\left(n(w) - \frac{1}{2}\right) - U_k .
\end{equation}
By choosing $v$ and $w$ to be two of the modes $i$, $j$, we see that the above class also contains attractive interactions
\begin{equation}\label{eq:density_density_interactions}
    V = - \sum_{i\,<\,j}\, |U_{ij}| \left(n_i -\frac{1}{2}\right)\left(n_j -\frac{1}{2}\right)
\end{equation}
with arbitrary strengths $|U_{ij}|$. 

Define the partition function
\begin{equation}
    \label{eq:partition_function}
    \ZZ(h) := \mathrm{Tr}\left[\exp\left(-\beta\left(H + h \sum_i n_i\right)\right)\right]
\end{equation}
as a function of $h \in \mathbb{C}$, where $n_i := a_i^\dagger a_i$.
\begin{theorem}
    \label{thm:Lee-Yang_theorem}
    For any Hamiltonian $H$ from the class in \cref{def:lee-yang}, in the complex plane of $h$, all the zeros of the partition function $\ZZ(h)$ satisfy $ \mathrm{Re}(h) \in [-\lambda_{\max}(A), 0].$
    In particular, when $A = 0$, they all lie on the imaginary axis.
\end{theorem}

We first prove the theorem for the Trotterized version of the partition function, as is commonly done, and then extend it to the exact, non-Trotterized version. Because the trace is unitarily invariant, the fermionic Lee-Yang theorem naturally applies to $\mathcal{U}^\dagger H \mathcal{U}$ for any Bogoliubov transformation $\mathcal{U}$, so long as the external field $\sum_i n_i \to \mathcal{U}^\dagger \left(\sum_i  n_i\right) \mathcal{U}$ is rotated appropriately. This extended class includes many well-studied models of interacting fermions, which we discuss in \cref{sec:physical_models}. The class also contains, under a two-mode encoding, the quantum spin Lee-Yang Hamiltonians \cite{suzuki1971}, as we show in \cref{sec:spin_hamiltonians}.

\paragraph{Efficient quantum algorithm for ground energy estimation.}
An immediate algorithmic implication of our Trotterized Lee-Yang theorem for fermions is a field-induced spectral gap for this Hamiltonian class. This consequently implies that, on a quantum computer, one can efficiently estimate their ground-state energies, essentially following the arguments from \cite{rayudu2026spectral}.
\begin{theorem}
\label{thm:spectral_gap}
    For any Hamiltonian $H$ from the class in \cref{def:lee-yang}, $H + h \sum_i n_i$ with $h>0$ has a spectral gap of at least $h/8$ above the ground state. Consequently, there is an efficient quantum algorithm, running in time $\mathrm{poly}(N, \|H\|_{\mathrm{op}}, 1/\epsilon)$, to estimate the ground-state energy of $H$ to within any $\epsilon > 0$ additive precision, where $N$ is the number of fermion modes.
\end{theorem}

Before we prove these statements in \cref{sec:free-fermions,sec:interacting}, we discuss further consequences of this fermionic Lee-Yang theorem.

\subsection{\label{sec:physical_models}Physical models}

The class of fermionic Lee-Yang Hamiltonians contains, up to unitary equivalence, several models of physical interest. In particular, despite the attractive nature of the interactions in \cref{def:lee-yang} and the PSD restriction on the hopping matrix, simple transformations of the modes enable us to go significantly beyond such seeming constraints. 
We give a few examples based on Hubbard models, which are central to the study of strongly correlated electrons \cite{martin2016interacting,arovas2022hubbard}. 

\paragraph{Attractive Hubbard model in a pairing field.}
Consider the spinful attractive Fermi-Hubbard model,
\begin{align}\label{eq:attractive_hubbard}
    H \,=\, \sum_{\sigma \in \{\uparrow, \downarrow\}}\,\sum_{i,j}\, t_{ij, \sigma}\; a^{\dagger}_{i\sigma} a_{j\sigma} \,-\, \sum_{i}\, |U_i| \left(n_{i\uparrow} - \frac{1}{2}\right)\left(n_{i\downarrow} - \frac{1}{2}\right),
\end{align}
where the hopping matrices satisfy $t_{ij,\uparrow} = t_{ij,\downarrow}^*$.
To see that \cref{thm:Lee-Yang_theorem} applies to this family, we apply the following canonical transformation (known as a pseudospin rotation \cite{anderson1958random}):
\begin{align}
\label{eq:canonical_transformation_attractive_Hubbard_model}
    b_{i1} \,:=\, \frac{a_{i\uparrow} + a^{\dagger}_{i\downarrow}}{\sqrt{2}} \quad \text{and} \quad b_{i2} \,:=\, \frac{a^{\dagger}_{i\uparrow} - a_{i\downarrow}}{\sqrt{2}}.
\end{align}
The condition $t_{ij,\uparrow}=t_{ij,\downarrow}^{*}$ ensures that all number-conserving quadratic terms cancel under this transformation, therefore mapping the model's hopping terms into purely pairing terms up to an additive constant.
In the context of \cref{eq:free_fermion}, this leaves us with $A = 0$ and nonzero $B$. Meanwhile, the attractive interactions remain attractive under this transformation.

Most interesting is what happens to the external field: for each $i$, the pseudospin rotation maps
\begin{equation}
    \tilde{n}_{i1} + \tilde{n}_{i2} = a^{\dagger}_{i\uparrow}a^{\dagger}_{i\downarrow} + a_{i\downarrow}a_{i\uparrow} + 1
\end{equation}
where $\tilde{n}_{i\alpha} = b_{i\alpha}^\dagger b_{i\alpha}$ are the occupation-number operators in the pseudospin basis. Thus the field that gives rise to the Lee-Yang phenomenon here is in fact a \textit{pairing} field in the original basis. 

We can generalize the above class to include some forms of spin-orbit couplings. Consider hopping terms of the form
\begin{equation}\label{eq:hopping_spin_orbit}
    \sum_{i,j} \sum_{\sigma, \sigma'} t_{i\sigma,\, j\sigma'}\, a^{\dagger}_{i\sigma} a_{j\sigma'}
\end{equation}
with the constraint that $t$ is time-reversal invariant, i.e., $J^T t J = t^{*}$ where $J = \I \otimes i\sigma_y$. This is equivalent to the matrix $t$ being of the form $C \otimes\I_2 \,+\, i\sum_{\alpha\in x,y,z} D_\alpha \otimes \sigma_{\alpha}$ with real symmetric $C$ and real antisymmetric $D_\alpha$. These terms also map to purely pairing Hamiltonians under the pseudospin rotation, and therefore the Lee-Yang theorem applies with the pairing field.

\paragraph{Repulsive Hubbard model on bipartite graphs.}
Another family consists of spinful \textit{repulsive} Hubbard models on any bipartite graph. Consider the vertex partitioning $\mathcal{A} \cup \mathcal{B}$ and take the hopping terms across the bipartition to be spin independent:
\begin{align}\label{eq:repulsive_bipartite_hubbard}
    H = \sum_{\sigma \in \{\uparrow, \downarrow\}}\,\sum_{i \in \mathcal{A},\, j \in \mathcal{B}}\, \left( t_{ij}\; a^{\dagger}_{i\sigma} a_{j\sigma} +t^{*}_{ij}\; a^{\dagger}_{j\sigma} a_{i\sigma} \right) \,+\, \sum_{i}\, |U_i| \left(n_{i\uparrow} - \frac{1}{2}\right)\left(n_{i\downarrow} - \frac{1}{2}\right).
\end{align}
To exhibit a Lee-Yang phenomenon, we place this system in a staggered Zeeman field 
\begin{equation}
    h \sum_{i \in \mathcal{A}} (n_{i\uparrow} - n_{i\downarrow}) - h \sum_{i \in \mathcal{B}} (n_{i\uparrow} - n_{i\downarrow}).
\end{equation}

To show that this family also lies in the fermionic Lee-Yang class, we may apply a simple unsigned particle-hole transformation ($a \rightarrow a^{\dagger}$ and $a^{\dagger} \rightarrow a$) on the modes $(i, \downarrow)$ for all $i \in \mathcal{A}$, and $(i, \uparrow)$ for all $i \in \mathcal{B}$; the other modes remain untouched. Under this mapping, the hopping terms become pairing, the repulsive interaction becomes attractive, and the staggered field becomes uniform, thereby meeting the criteria for \cref{thm:Lee-Yang_theorem} to hold.

The theorem generalizes to spinless repulsive bipartite Hubbard models as well.
Take
\begin{align}
     H = \sum_{i \in \mathcal{A},\,j \in \mathcal{B}}\, \left( t_{ij}\; a^{\dagger}_{i} a_{j} +t^{*}_{ij}\; a^{\dagger}_{j} a_{i} \right) \,+\, \sum_{i \in \mathcal{A},\,j \in \mathcal{B}}\, |U_{ij}| \left(n_{i} - \frac{1}{2}\right)\left(n_{j} - \frac{1}{2}\right)
\end{align}
in a staggered field $h \sum_{i} \eta_i n_i$, where $\eta_i = 1$ for $i \in \mathcal{A}$ and $\eta_i = -1$ for $i \in \mathcal{B}$. A particle-hole transformation on sublattice $\mathcal{B}$ transforms the Hamiltonian into a subclass of \cref{eq:free_fermion_and_interactions} as before, and maps the staggered field into a uniform field.

\paragraph{Doped Hubbard model.} 
The Hamiltonian in \cref{eq:repulsive_bipartite_hubbard} has no chemical potential, so the system is at half filling. One can also dope the system by adding a chemical potential $-\mu \sum_i (n_{i\uparrow} + n_{i\downarrow})$. Under the same particle-hole transformation, this term becomes a staggered Zeeman field, which is number conserving. In terms of \cref{eq:free_fermion}, this leaves us with a nonzero $A$ whose eigenvalues are $\pm\mu$. Shifting $A$ by $|\mu|$ makes it PSD and shifts the uniform field by the same amount, so \cref{thm:Lee-Yang_theorem} confines the zeros of the partition function of the doped model in the staggered field to
\begin{equation}
    |\mathrm{Re}(h)| \,\leq\, |\mu| ,
\end{equation}
and \cref{thm:spectral_gap} gives a spectral gap of at least $(h - |\mu|)/8$ whenever $h > |\mu|$. The doped Hubbard model is the setting relevant to high-temperature superconductivity, where, unlike at half filling, quantum Monte Carlo simulations suffer from the fermionic sign problem (see \cref{sec:sign_problem}) and the ground state remains the central open problem of correlated-electron physics \cite{PhysRevX.5.041041,zheng2017stripe,qin2022hubbard}. Our work therefore provides a concrete analytic result for the doped Hubbard model, which is otherwise beyond rigorous control; however, it holds only when the external field strength is larger than the doping, $h > |\mu|$, which is outside the regime of primary physical interest.

\paragraph{Repulsive Hubbard model with large magnetic field.}
One can apply the particle-hole transformation in a different way to obtain repulsive Hubbard models even when the graph is nonbipartite. 
Such a Hamiltonian satisfies the fermionic Lee-Yang condition as long as a sufficiently strong Zeeman field is added: 
\begin{align}
H = \sum_{i,j,\sigma} t_{ij} a_{i\sigma}^\dagger a_{j\sigma}
+ \sum_i |U_i| \left(n_{i\uparrow}-\frac{1}{2}\right)\left(n_{i\downarrow}-\frac{1}{2}\right)+h_{\mathrm{Zeeman}}\sum_i\left(n_{i\uparrow}-n_{i\downarrow}\right),
\end{align} 
subject to a field $h\sum_i (n_{i\uparrow}-n_{i\downarrow})$ in the same direction as $h_{\mathrm{Zeeman}}$. 
With the particle-hole transformation only applied to the $\sigma={\downarrow}$ modes for \textit{all} sites, the repulsive interaction turns attractive, and the hopping term for $\sigma={\uparrow}$ modes remains unchanged. 
The other hopping terms obtain a minus sign $t_{ij} a^\dagger_{i\downarrow}a_{j\downarrow} \to t_{ij} a_{i\downarrow}a_{j\downarrow}^\dagger = -t_{ij} a_{j\downarrow}^\dagger a_{i\downarrow}$, and therefore the corresponding matrix $A$ from \cref{eq:free_fermion_and_interactions} is block diagonal with blocks $t+h_{\mathrm{Zeeman}}\I$ and $-t^{T}+h_{\mathrm{Zeeman}}\I$.
Observe that $t$ and $-t^{T}$ have the same eigenvalues but with flipped signs, so in order to keep $A$ PSD one needs a field strength $h_{\mathrm{Zeeman}}\geq \|t\|_{\mathrm{op}}$.

Although this example has a fully polarized product ground state, its inclusion highlights the unifying scope of the fermionic Lee-Yang framework. 
The half-filled repulsive Hubbard model on a bipartite lattice admits sign-problem-free quantum Monte Carlo formulations, despite its generically nontrivial ground state (we elaborate further on this in the next section). In contrast, on a non-bipartite graph, a sufficiently strong uniform Zeeman field makes the ground state trivial \cite{StrackVollhardt1994PRL} without necessarily eliminating the sign problem in conventional auxiliary-field formulations \cite{Iglovikov2015PRB}. 
These examples are thus both computationally \textit{accessible} (although not necessarily provably efficient) for quite different reasons. 
The fermionic Lee-Yang framework places them within a common analytic class, revealing a shared structure beyond the case-by-case arguments underlying their numerical or analytical accessibility.

\subsection{\label{sec:spin_hamiltonians}Relation to quantum spin Lee-Yang Hamiltonians}

In this section, we show that the Lee-Yang theorem for spin Hamiltonians can be subsumed by the fermionic Lee-Yang theorem. Consider the spin Hamiltonian class:
\begin{align}\label{eq:spin_lee_yang}
    H_{\mathrm{s}} \,=\, \sum_{i<j} \left( -\,w^{z}_{ij}\, Z_i Z_j \,+\, w^{x}_{ij}\, X_i X_j \,+\, w^{y}_{ij}\, Y_i Y_j \,+\, w^{xy}_{ij}\, X_i Y_j \,+\, w^{yx}_{ij}\, Y_i X_j \right),
\end{align}
where, for every pair of qubits $i,j$, the coefficients satisfy
\begin{align}
    \label{eq:spin_LY_condition}
    w^z_{ij} \;\geq\; \big\| \mathsf{w}_{ij} \big\|_{\mathrm{op}}, \qquad \mathsf{w}_{ij} \,:=\, \begin{pmatrix} w^{x}_{ij} & w^{xy}_{ij} \\[2pt] w^{yx}_{ij} & w^{y}_{ij} \end{pmatrix}.
\end{align}
Explicitly,
\begin{align*}
    \big\| \mathsf{w}_{ij} \big\|_{\mathrm{op}} \,=\, \frac{1}{2} \Big[ \big( w^x_{ij} - w^y_{ij} \big)^2 + \big( w^{xy}_{ij} + w^{yx}_{ij} \big)^2 \Big]^{\frac{1}{2}} + \frac{1}{2} \Big[ \big( w^x_{ij} + w^y_{ij} \big)^2 + \big( w^{xy}_{ij} - w^{yx}_{ij} \big)^2 \Big]^{\frac{1}{2}}.
\end{align*}
These Hamiltonians have the Lee-Yang property \cite{suzuki1971} with respect to the longitudinal field $h\sum_i Z_i$. We show that every such $H_{\mathrm{s}}$ can be encoded into a fermionic Lee-Yang Hamiltonian from \cref{def:lee-yang} and we can recover the spin Lee-Yang theorem from the fermionic Lee-Yang theorem. In this sense, the fermionic class strictly contains the spin class.

For each qubit $i$, consider a pair of fermion modes $(i,\uparrow)$ and $(i,\downarrow)$. Following \cite{bravyi2002fermionic}, we encode the qubit into the empty and the doubly occupied states of the modes. We can then write the Pauli operators as
\begin{align}\label{eq:Pauli_encoding}
    Z_i \,\coloneq\, n_{i\uparrow} + n_{i\downarrow} - \I, \quad X_i \,\coloneq\, a^{\dagger}_{i\uparrow} a^{\dagger}_{i\downarrow} + a_{i\downarrow} a_{i\uparrow}, \quad Y_i \,\coloneq\, -i \left( a^{\dagger}_{i\uparrow} a^{\dagger}_{i\downarrow} - a_{i\downarrow} a_{i\uparrow} \right),
\end{align}
and the uniform field $h\sum_i Z_i$ as $h\sum_{i\sigma} n_{i\sigma}$ up to an additive constant. The remaining singly occupied states are separated from the encoded subspace by adding an attractive on-site interaction
\begin{align}\label{eq:encoding_penalty}
    \lambda \sum_i \left(n_{i\uparrow} - n_{i\downarrow}\right)^2
\end{align}
for $\lambda > 0$, which vanishes on the encoded subspace. Note that this type of attractive interaction is of the form \cref{eq:density_density_interactions} up to an additive constant and is therefore allowed in the fermionic Lee-Yang class.

\begin{table}[ht]
    \centering
    \renewcommand{\arraystretch}{1.4}
    \begin{tabular}{@{}lll@{}}
        \toprule
        Terms in \cref{eq:interactions} & Modes & Restriction to the encoded subspace \\
        \midrule
        $-\,2\big(q^{+}_{v,w}\, q^{-}_{v,w} + q^{-}_{v,w}\, q^{+}_{v,w}\big)$
        & $v = (i,\uparrow)$, $w = (j,\uparrow)$
        & $-\,Z_i Z_j$ \\
        \addlinespace[10pt]
        $-\,4\big(q^{+}_{v,w}\, q^{-}_{v,w} + q^{-}_{v,w}\, q^{+}_{v,w}\big)$
        & \begin{tabular}[c]{@{}l@{}}
            $a^{\dagger}(v) = \frac{1}{\sqrt{2}}\big(a^{\dagger}_{i\uparrow} \pm a^{\dagger}_{j\uparrow}\big)$ \\[2pt]
            $a^{\dagger}(w) = \frac{1}{\sqrt{2}}\big(a^{\dagger}_{i\downarrow} + a^{\dagger}_{j\downarrow}\big)$
          \end{tabular}
        & $-\,\left(Z_i Z_j \,\pm\, X_i X_j \,\pm\, Y_i Y_j\right)$ \\
        \addlinespace[10pt]
        $-\,4\big(q^{+}_{v,w}\, q^{\pm}_{v',w'} + q^{\mp}_{v',w'}\, q^{-}_{v,w}\big)$
        & \begin{tabular}[c]{@{}l@{}}
            $v = (i,\uparrow)$, $w = (j,\uparrow)$ \\[2pt]
            $v' = (i,\downarrow)$, $w' = (j,\downarrow)$
          \end{tabular}
        & $-\,\left(Z_i Z_j \,\mp\, X_i X_j \,\pm\, Y_i Y_j\right)$ \\
        \bottomrule
    \end{tabular}
    \caption{\label{tab:encoded_interactions} Fermionic Lee-Yang interactions and their restrictions to the encoded subspace.}
\end{table}

On the encoded subspace, the three fermionic interactions in \cref{tab:encoded_interactions} produce the couplings listed in the rightmost column of the table. The convex combinations of those three interactions produce every coupling $-\,w^{z}_{ij}\, Z_i Z_j + (X_i, Y_i)\, \mathsf{w}_{ij}\, (X_j, Y_j)^{T}$ between qubits $i$ and $j$ in \cref{eq:spin_lee_yang} with diagonal $\mathsf{w}_{ij}$ satisfying \cref{eq:spin_LY_condition}. Rotating qubit $i$ about the $Z$-axis by an angle $\alpha$ and qubit $j$ by $\beta$ maps $\mathsf{w}_{ij}$ to $R_\alpha\, \mathsf{w}_{ij}\, R_\beta^{T}$, with $R_\theta$ the rotation matrix of the $X$-$Y$ plane in the Bloch sphere. These rotations can be implemented by multiplying the modes $(i,\uparrow)$ and $(j,\uparrow)$ by phases in the last two rows of \cref{tab:encoded_interactions}.
Since this rotation does not change the operator norm of $\mathsf{w}_{ij}$, the rotated $\mathsf{w}_{ij}$ still satisfies \cref{eq:spin_LY_condition}. Collecting the interactions used for every pair $i, j$ gives an interaction $V$ of the form \cref{eq:interactions}, whose restriction to the encoded subspace is $H_{\mathrm{s}}$ up to an additive constant.

Finally, to obtain the zero-freeness theorem restricted to the encoded space, let $V_\lambda$ be $V$ with the terms in \cref{eq:encoding_penalty} added, which is again of the form \cref{eq:interactions} up to an additive constant. As $\lambda \to \infty$, the states outside the encoded subspace drop out of the partition function of $V_\lambda + h\sum_i n_i$, which converges, locally uniformly in $h$, to the partition function of $H_{\mathrm{s}} + h\sum_i Z_i$ up to a nonvanishing factor. \cref{thm:Lee-Yang_theorem} gives the Lee-Yang theorem for the former, and Hurwitz's theorem gives the same for the latter.

\subsection{\label{sec:sign_problem}Relation to the fermionic sign problem}

It is well known that the attractive Hubbard and repulsive bipartite (half-filled) Hubbard models, assuming certain symmetries on the hopping matrix $t$, admit a sign-problem-free (SPF) representation through a quantum Monte Carlo method called determinant Monte Carlo \cite{blankenbecler1981monte,hirsch1983discrete}. This SPF property alone does not imply efficient classical simulability, and it is also understood differently from the usual notion of stoquasticity in quantum spin Hamiltonians \cite{bravyi2006complexity}. Instead, it arises from applying the Hubbard-Stratonovich transformation to a Trotterized partition function,
\begin{equation}\label{eq:hubbard-stratonovich}
    \ZZ = \sum_{\phi} w(\phi) \, \mathrm{Tr}\left[ \, \prod_{\ell} e^{K_\ell(\phi)} \, \right], \quad w(\phi) \geq 0.
\end{equation}
Here, $\phi$ runs over configurations of classical spins called an auxiliary field, and each exponential is a fermionic Gaussian operator generated by some
\begin{equation}\label{eq:K_generator}
    K_\ell(\phi) := \sum_{ij} [A_\ell(\phi)]_{ij} a_i^\dagger a_j + \frac{1}{2} \sum_{ij} \left( [B_\ell(\phi)]_{ij} a_i^\dagger a_j^\dagger + [B_\ell(\phi)]_{ij}^* a_j a_i \right) - \frac{1}{2} \mathrm{Tr}\,A_\ell(\phi)
\end{equation}
for some $\phi$-dependent complex matrices $A_\ell(\phi), B_\ell(\phi)$. In other words, these methods reduce the partition function of an interacting fermionic system to a classical sum, which may then be estimated via quantum Monte Carlo. The weights are given by tracing out free-fermion operators, which can be computed efficiently as determinants involving the $A$- and $B$-type matrices. The SPF property, then, corresponds to when all such determinants are nonnegative.

There are numerous sufficient, often overlapping, conditions that yield such an SPF representation \cite{li2019sign}. The class of Hamiltonians in \cref{def:lee-yang} turns out to have a nontrivial intersection with these conditions, and we demonstrate this with two broadly encapsulating examples: Majorana reflection positivity \cite{wei2016majorana} and Kramers time-reversal symmetry \cite{wu2005sufficient}. Below, we drop the dependence on the auxiliary field $\phi$ and the Trotter layer $\ell$, as the conditions we discuss are demanded to hold for every $\phi$ and $\ell$.

\paragraph{Majorana reflection positivity.}
One broad SPF condition was identified in Ref.~\cite{wei2016majorana}, encompassing both the attractive and repulsive bipartite Hubbard models with real hopping. Known as Majorana reflection positivity (MRP), it had subsumed all SPF fermionic Hamiltonians that were known at the time of its introduction, including the split orthogonal group class \cite{wang2015split} and the Majorana class within Majorana time-reversal (MTR) symmetry \cite{li2015solving,li2016majorana}.

As the name implies, MRP is naturally phrased in terms of Majorana operators; we re-express the condition in terms of the ladder operators, keeping in mind that there is a Bogoliubov degree of freedom in how to define the Majorana fermions from the Dirac fermions $a_i, a_i^\dagger$. While MRP does not require the generator $K$ to be Hermitian, we assume so here for simplicity. If all matrices $A, B$ appearing in the decomposition of $\ZZ$ obey the following properties:
\begin{equation}\label{eq:MRP_condition_intro}
    A = A^T \text{ real with either } A \succeq 0 \text{ or } A \preceq 0, \qquad B = -B^T \text{ imaginary},
\end{equation}
then the Hamiltonian satisfies MRP and is therefore SPF in the sense of \cref{eq:hubbard-stratonovich} \cite{wei2016majorana}. Comparing this to the constraints from \cref{def:lee-yang}, we see that MRP intersects with the quadratic part $H_0$ of the fermionic Lee-Yang class at a purely real/purely imaginary slice. It is particularly remarkable that the PSD conditions between the two coincide (note that the Lee-Yang PSD condition can be flipped to negative semidefinite by a simple basis change). 

\paragraph{Kramers symmetry.}
One may observe that the attractive Hubbard model with complex hoppings evades \cref{eq:MRP_condition_intro}, yet those satisfying $t_{ij,\uparrow} = t_{ij,\downarrow}^*$ easily yield an SPF decomposition. This is because the SPF mechanism there is Kramers time-reversal symmetry \cite{wu2005sufficient}. The generalization to spin-orbit couplings, as in \cref{eq:hopping_spin_orbit} with time-reversal symmetry, also lies in this Kramers class. Suppose the quadratic operator $K$ is particle-preserving ($B = 0$); we say $K$ is Kramers-invariant if there exists an antiunitary operator $\Theta$ with $\Theta^2 = -\I$ such that $\Theta A \Theta^{-1} = A$. Thus all such $A \succeq 0$ in the fermionic Lee-Yang class obey Kramers symmetry.

We can remove the PSD constraint by using the pseudospin rotation again. Suppose the number of modes $N$ is even. We can show that the subfamily of \cref{def:lee-yang} such that 
\begin{equation}
    A = 0, \qquad B = JB^*J^T
\end{equation}
yields a Kramers-invariant $H_0$, where $J = \I \otimes i\sigma_y$. For each pair of modes $j \in [N/2]$, let
\begin{equation}
b_{2j-1}=\frac{a_{2j-1}-a_{2j}^\dagger}{\sqrt2},
\qquad
b_{2j}=\frac{a_{2j-1}^\dagger+a_{2j}}{\sqrt2}.
\end{equation}
Substituting this into the quadratic part of \cref{def:lee-yang} yields
\begin{equation}
    H_0 = \lambda + \sum_{i,j} h_{ij} b_i^\dagger b_j, \qquad \lambda = -\frac{1}{2} \mathrm{Tr}(h), \qquad h = BJ.
\end{equation}
Indeed, $h$ here is Hermitian but not necessarily PSD. In particular, it obeys Kramers symmetry, which can be seen through the antiunitary $\Theta = J\mathcal{K}$ where $\mathcal{K}$ denotes complex conjugation in the standard basis. Then $\Theta^2 = -\I$ and $\Theta h \Theta^{-1} = h$, which is precisely the SPF condition of \cite{wu2005sufficient}.

Note that the Kramers mechanism has been extended to non-particle-conserving systems, known as the Kramers class of MTR \cite{li2016majorana}. This is distinct from MRP; in contrast, the so-called Majorana class of MTR \cite{li2015solving} is subsumed by MRP. Thus nontrivial pairing terms from our fermionic Lee-Yang class can in principle support Kramers MTR as well.

\paragraph{Attractive interactions.}
Fermionic interactions are not quadratic, and so do not immediately fall under \cref{eq:K_generator}. The standard approach is to apply the Hubbard-Stratonovich transformation to represent their exponentials as a weighted sum of Gaussian operators, thereby yielding the form of \cref{eq:hubbard-stratonovich}. We review how this works using the special case \cref{eq:density_density_interactions} of density-density interactions with attractive signs. For any $t \geq 0$ we have the identity~\cite{hirsch1983discrete}
\begin{equation}
    e^{t |U_{ij}| ( n_{i} - \frac{1}{2} ) ( n_{j} - \frac{1}{2} )} = \frac{1}{2} e^{-t |U_{ij}|/4} \sum_{\phi \in \{\pm 1\}} e^{t \lambda \phi (n_{i} + n_{j} - 1)}, \qquad \cosh(\lambda t) = e^{t|U_{ij}|/2}.
\end{equation}
This is a classic example admitting an SPF representation. We can see this through MRP:~the corresponding matrices are $A = t\lambda\phi(\ketbra{i}{i} + \ketbra{j}{j})$ and $B = 0$. Indeed, $A$ is purely real and $A \succeq 0$ for $\phi = +1$ ($A \preceq 0$ for $\phi = -1$), while $B = 0$ is purely imaginary.

\paragraph{Potential for quantum advantage.}
Crucially, there are many fermionic Lee-Yang Hamiltonians for which no SPF representation is known. For example, this occurs in repulsive bipartite Hubbard models with complex hopping amplitudes $t_{ij}$. A canonical example is the interacting Hofstadter model \cite{tai2017microscopy,andrews2018stability}, which at half filling is a spinless version of \cref{eq:repulsive_bipartite_hubbard} with $t_{ij} = -e^{i \theta_{ij}}$. The phases $\theta_{ij}$ are chosen to satisfy some conserved flux relation; for example on a square lattice, we may have that the phases in each plaquette sum to $2\pi/3$. These nontrivial phases result in complex-valued determinants appearing in $\ZZ$, thereby obstructing any known route to an SPF representation. As such, modern investigations of interacting Hofstadter models typically employ either DMRG methods, which even in the gapped phase exhibits exponential computational cost in the width of the two-dimensional lattice \cite{motruk2016density}, or otherwise exact diagonalization on small sizes \cite{funfhaus2024topological}. In contrast, by \cref{thm:spectral_gap} we can now show that there is a quantum algorithm to study such models with only polynomial cost.

\section{\label{sec:free-fermions}Lee-Yang theorem for free fermions}
We first prove a Lee-Yang theorem for free-fermion Hamiltonians $H_0$ with $A \succeq 0$, as defined in \cref{eq:free_fermion}. Then in \cref{sec:interacting} we will extend the theorem to include the interactions. In both cases, the theorem is proved with respect to a uniform onsite potential $h\sum_{i} n_i$. Similar to the circle theorem for spin Hamiltonians \cite{asano1970,suzuki1971}, the theorem for fermions is derived for the Trotterized partition function.

Let $N$ be the total number of fermionic modes, $M$ the number of Trotter layers, and define $\epsilon := \beta/M$. Consider the following polynomial
\begin{align}
    \label{eq:free_fermion_partition_function}
    \ZZ_{\epsilon}(H_0, \bm{u}) := \mathrm{Tr}\left[\,\prod_{m\in \Z_M} \, \left[e^{-\epsilon H_0}\, \prod_{i \in \Z_N} \left(\I-n_i + u_{i,m}\,n_i\right) \right]\,\right] = \mathrm{Tr}\left[\,\prod_{m} \, \left[e^{-\epsilon H_0}\, \prod_{i} e^{n_i 
\log(u_{i,m})} \right]\,\right]
\end{align}
as a function of the complex variables $\{u_{i,m}: i\in \Z_N, m\in \Z_M\}$. At $\bm{u} = e^{-\epsilon h} \bm{1},$ the above polynomial evaluates to the Trotterized partition function of $H_0 +h\sum_{i} n_i$:
\begin{align}
    \ZZ_\epsilon(H_0, \bm{u} = e^{-\epsilon h} \bm{1}) \,&=\, \mathrm{Tr}\left[\,\prod_{m} \, \left[e^{-\epsilon H_0}\, \prod_{i} e^{-\epsilon h n_i} \right]\,\right] \,\approx \, \mathrm{Tr}\left[\exp\left(-\beta\bigg(H_0 +h\sum_{i} n_i \bigg)\right)\right]
\end{align}
where the approximation is controlled by Trotter error bounds \cite{childs2021}. In the rest of the paper, we also refer to the polynomials of the form $\ZZ_\epsilon(H_0, \bm{u})$ as Trotterized partition functions.

\begin{theorem}
    \label{thm:zero-free_free_fermion}
    In the open unit polydisc, i.e., $|u_{i,m}| < 1$ for all $i$ and $m$, the Trotterized partition function $\ZZ_{\epsilon}(H_0, \bm{u})$, defined in \cref{eq:free_fermion_partition_function}, is nonzero for any free-fermion Hamiltonian $H_0$ defined in \cref{eq:free_fermion_and_interactions}. 
\end{theorem}
\begin{proof}
    Let 
    \begin{align}
        Q :=
        \begin{pmatrix}
            A & B \\
            B^{\dagger} & -A^*
        \end{pmatrix}
        \quad \text{such that} \quad Q = Q^{\dagger},
    \end{align}
    and rewrite the free-fermion Hamiltonian from \cref{eq:free_fermion} as
    \begin{align}
        H_0 = \frac{1}{2}\bm{\psi}^{\dagger} Q \bm{\psi} + \frac{1}{2} \mathrm{Tr}(A) \quad \text{where} \quad \bm{\psi}^\dagger := (a_0^\dagger,\dots,a_{N-1}^\dagger,a_0,\dots,a_{N-1}).
    \end{align}
    Let $W := e^{-\epsilon Q}$, $E_m \,:=\, \diag\big(\{u_{i,m}\}_{i}, \bm{1}_N\big)$, $F_m \,:=\, \diag\big(\bm{1}_N, \{u_{i,m}\}_{i}\big)$, and $G_m := E_m F^{-1}_m$. Rewrite the square of the partition function as
    \begin{align}
        \ZZ_\epsilon(H_0,\bm{u})^2 &\,=\, e^{-\beta \,\mathrm{Tr}(A)}\left(\prod_{i,m} u_{i,m} \right)\mathrm{Tr}\left[\rule{0cm}{0.85cm}\,\prod_{m}\, \left[e^{- \frac{1}{2}\epsilon\,\bm{\psi}^{\dagger} Q \bm{\psi}}\, e^{\frac{1}{2}\bm{\psi}^{\dagger} \log(G_m) \bm{\psi}} \right]\,\right]^2 \\[6pt]
        &\,=\, e^{-\beta \,\mathrm{Tr}(A)}\left(\prod_{m} \det(F_m) \right) \det\left[\I + \prod_{m}W G_{m}\right].
    \end{align} 
    The determinant of $\I$ plus a product of matrices can be written as the determinant of a single block matrix\footnote{This follows from Schur complement: the top-left $(M-1)\times(M-1)$ block of $\mathfrak{B}$ is block-triangular with unit diagonal, and eliminating it leaves exactly $\det[\I + \prod_{m \in \Z_M}WG_m]$.}
    \begin{align}
        \det\left[\I + \prod_{m}W G_m\right] \,=\, \det \mathfrak{B} \quad \text{where} \quad \mathfrak{B} \,:=\,
        \begin{pmatrix}
            \I      & -W G_1 &        &       \\
                    & \I     & \ddots &      \\
                    &        & \ddots & -W G_{M-1}     \\
            W G_0   &        &        & \I
        \end{pmatrix},
    \end{align}
    with all the unspecified blocks equal to zero. Multiplying with $\prod_{m} \det(F_m)$, we get
    \begin{align}
        \label{eq:pencil_representation}
        \ZZ_\epsilon(H_0,\bm{u})^2  = e^{-\beta \,\mathrm{Tr}(A)}\det \mathfrak{D} \quad \where \quad 
        \mathfrak{D} = \begin{pmatrix}
            F_0     & -W E_1 &        &       \\
                    & F_1    & \ddots &      \\
                    &        & \ddots & -W E_{M-1}     \\
            W E_0   &        &        & F_{M-1}
        \end{pmatrix}.
    \end{align}
    Although the derivation above assumes $u_{i,m}\neq 0$ for all $i,m$,
    Eq.~\eqref{eq:pencil_representation} remains valid when some variables vanish,
    since both sides are polynomials in $\bm{u}$.

    We now show that the matrix $\mathfrak{D}(\bm u)$ is non-singular everywhere in the open unit polydisc. Suppose $\mathfrak{D}(\bm u)$ is singular and that $\phi = (\phi_0, \dots, \phi_{M-1})$ is a corresponding null vector. Then the structure of the matrix implies
    \vspace{-0.3cm}
    \begin{align}
        \label{eq:balance_equation}
        F_{m}\phi_{m} \,=\, \pm\, W E_{m+1}\, \phi_{m+1} \quad \text{for all }\, m
    \end{align}
    where the index $m$ is taken mod $M$ and the minus sign occurs only at $m = M-1$. Define a diagonal matrix $\Gamma := \diag( +\bm{1}_N,\, -\bm{1}_N)$, which satisfies $\{\Gamma, Q\} = 2\diag(A, A^*) \succeq 0$.
    Since $\partial_{\epsilon} W\Gamma W = -W\{\Gamma,Q\}W \preceq 0$ for all $\epsilon \in \R$, this implies $W \Gamma W \preceq \Gamma$. With $\Gamma$, define a quadratic form
    \begin{align}
        \braket{\psi, \psi}_{\Gamma} \,:=\, \braket{\psi, \Gamma \psi} \,=\, \sum_{j\in \Z_N} |\psi_j|^2 - \sum_{j\in \Z_{2N}\backslash \Z_N} |\psi_j|^2.
    \end{align}
    The quadratic form is monotonically non-increasing under the action of $W$ since
    \begin{align}
        \braket{W\psi, W\psi}_{\Gamma} \,=\,  \braket{W\psi, \Gamma W\psi} \,=\, \braket{\psi, W \Gamma W \psi} \,\leq\, \braket{\psi, \Gamma\psi} = \braket{\psi,\psi}_{\Gamma},
    \end{align} 
    where we used $W^{\dagger} = W$ since $Q$ is Hermitian. Taking the quadratic form of both sides of Eq.~\eqref{eq:balance_equation}, the sign squares away and we get
    \begin{align}
        \sum_{j\in \Z_N} |\phi_{m,j}|^2 \,-\, \sum_{j\in \Z_{2N}\backslash\Z_N} |u_{j-N,m}|^2|\phi_{m,j}|^2 \;\leq\;  \sum_{j\in \Z_N} |u_{j,m+1}|^2|\phi_{m+1,j}|^2 \,-\, \sum_{j\in \Z_{2N}\backslash\Z_N} |\phi_{m+1,j}|^2
    \end{align}
    Putting these together by summing over $m \in \Z_M$ and shifting the summation index $m \to m-1$ on the right-hand side, all the terms combine into
    \begin{align}
        \sum_{m\in \Z_M}\, \sum_{j\in \Z_{2N}} \left(1 - |u_{j (\mathrm{mod}\ N),\ m}|^2\right) |\phi_{m,j}|^2 \;\leq\; 0.
    \end{align}
    In the open unit polydisc, every coefficient $1 - |u_{j,m}|^2$ is strictly positive, which implies $\phi_m = 0$ for all $m$, i.e., $\phi = 0$. Hence $\mathfrak{D}(\bm{u})$ is non-singular, and by Eq.~\eqref{eq:pencil_representation}, $\ZZ_\epsilon(H_0,\bm{u})^2 = e^{-\beta \,\mathrm{Tr}(A)}\det \mathfrak{D}(\bm{u}) \neq 0$ everywhere in the open unit polydisc.
\end{proof}

\begin{remark}\label{rem:contraction} 
The matrix argument in the above proof only requires $W^\dagger\Gamma W\preceq\Gamma$. Thus each Fock-space factor $e^{-\epsilon H_0}$ may be replaced by a product of exponentials of possibly non-Hermitian quadratic fermion operators which satisfy the above inequality. We will use this fact in \cref{sec:interacting} to extend the zero-freeness theorem to interacting Hamiltonians.  
\end{remark}

\section{\label{sec:interacting}Lee-Yang theorem for interacting fermions}

In this section, building on zero-freeness theorem from \cref{sec:free-fermions} and using Asano contractions \cite{asano1970,ruelle1971extension}, we extend the theorem from free fermions to the interacting Hamiltonians of \cref{def:lee-yang}. Let $\epsilon:=\beta/M$ and $T_k:=\I+\epsilon U_k r_k s_k$. Consider the following polynomial
\begin{align}
    \label{eq:interacting_fermion_partition_function}
    \ZZ_\epsilon(H,\bm u):=
    \mathrm{Tr}\left[\prod_{m\in\Z_M}\left[
    e^{-\epsilon H_0}
    \left(\prod_{k} T_k T_k^\dagger\right)
    \prod_i(\I-n_i+u_{i,m}n_i)
    \right]\right],
\end{align}
where
\begin{align}
    \prod_k T_k T_k^\dagger
    =\I+\epsilon\sum_k U_k(r_k s_k+s_k^\dagger r_k^\dagger)+O(\epsilon^2)
    =e^{-\epsilon V}+O(\epsilon^2).
\end{align}
So the polynomial in \cref{eq:interacting_fermion_partition_function}, with $\bm{u} = e^{-\epsilon h} \bm{1}$, is the Trotterized partition of $H +h\sum_{i} n_i$ where $H$ is the Hamiltonian with interactions from \cref{def:lee-yang}. 

To understand the zero-freeness of \cref{eq:interacting_fermion_partition_function}, we replace each occurrence of $T_k$ by
\begin{align}
    \label{eq:auxiliary_factors}
    \big[(1-y)\I+\sqrt{\epsilon U_k}(1+y)r_k\big]
    \big[(1-z)\I+\sqrt{\epsilon U_k}(1+z)s_k\big],
\end{align}
and $T_k^\dagger$ by a similar expression but with $s_k^\dagger,r_k^\dagger$ instead of $r_k,s_k$. Let us give each such replacement of $T_k$ and $T^{\dagger}_k$ at all the layers its own pair of variables $(y,z)$ and denote the resulting polynomial by $\widetilde{\ZZ}_\epsilon(\bm u,\bm y,\bm z)$. 

\begin{lemma}
    \label{lem:auxiliary_zero_free}
    In the open unit polydisc of the variables in $\bm u$, $\bm y$ and $\bm z$, the polynomial $\widetilde{\ZZ}_\epsilon(\bm u,\bm y,\bm z)$ is nonzero.
\end{lemma}
\begin{proof}
    By \cref{rem:contraction} and following the proof of \cref{thm:zero-free_free_fermion}, it suffices to express the additional insertions of the form in \cref{eq:auxiliary_factors} as nonzero scalars times quadratic exponentials $W$ satisfying $W^\dagger\Gamma W\preceq\Gamma$.

    For any $q\in\mathcal Q$, the identity $q^2=0$ gives
    \begin{align}
        (1-z)\I+\sqrt{\epsilon U_k}\,(1+z)q
        =(1-z)\,e^{xq}\quad \text{with} \quad x:=\sqrt{\epsilon U_k}\,\frac{1+z}{1-z}.
    \end{align}
    When $|z|<1$, the scalar $1-z$ is nonzero and $\mathrm{Re}(x)\geq0$. For $q=q^\sigma_{v,w}$ with $\sigma=\pm $, write $q=\frac12\bm\psi^\dagger\mathsf q\bm\psi$, where we use the convention $(v^\dagger ~ w^T)\bm\psi = a(v)+a^\dagger (w)$ and set 
    \begin{align}
        \mathsf q=\frac12\begin{pmatrix}
            -vv^\dagger-ww^\dagger & \sigma(vw^T-wv^T)\\
            \sigma(v^*w^\dagger-w^*v^\dagger) & v^*v^T+w^*w^T
        \end{pmatrix} \in\mathbb{C}^{2N\times2N}.
    \end{align}
    Orthonormality of $v$ and $w$ gives $\mathsf q^2=0$ and we can see $\Gamma\mathsf q=\mathsf q^\dagger\Gamma=-RR^\dagger/2\preceq0$ with 
    \begin{align}
    R=\begin{pmatrix}
        v & w\\
        -\sigma w^* & \sigma v^*
    \end{pmatrix}\in\mathbb{C}^{2N\times2} .
    \end{align}
    Hence the corresponding matrix $W=e^{x\mathsf q}=\I+x\mathsf q$ satisfies
    \begin{align}\label{eq:gammashrink}
        W^\dagger\Gamma W
        =\Gamma+2\mathrm{Re}(x)\Gamma\mathsf q\preceq\Gamma.
    \end{align}
    Since all the $r_k, s_k$ and their adjoints in $\widetilde{\ZZ}_\epsilon(\bm u,\bm y,\bm z)$ are elements in $\mathcal{Q}$, applying \cref{rem:contraction} proves the lemma.
\end{proof}

\begin{theorem}
    \label{thm:unit_disk_zero_free_interacting_fermions}
    For any Hamiltonian in \cref{def:lee-yang} and all sufficiently small $\epsilon>0$, the Trotterized partition function $\ZZ_\epsilon(H,\bm u)$, defined in \cref{eq:interacting_fermion_partition_function}, is nonzero throughout the open unit polydisc.
\end{theorem}
\begin{proof}
    Apply the Asano contraction
    \begin{align}
        a+by+cz+dyz\longmapsto a+dy
    \end{align}
    to every auxiliary pair in $\widetilde{\ZZ}_\epsilon$. By \cref{lem:auxiliary_zero_free} and the Asano contraction lemma \cite{asano1970,ruelle1971extension}, the contracted polynomial remains nonzero in the open unit polydisc.

    At the contracted variable $y=1$, each replacement in \cref{eq:auxiliary_factors} becomes
    \begin{align}
        2(\I+\epsilon U_k r_k s_k)=2T_k,
    \end{align}
    and each reversed adjoint replacement becomes $2T_k^\dagger$. Consequently, after dividing by $2$, letting the auxiliary variables approach one from within the unit disc gives $\ZZ_\epsilon(H,\bm u)$, locally uniformly in $\bm u$.

    Hurwitz's theorem gives zero-freeness unless this limit is identically zero. That alternative is excluded for sufficiently small $\epsilon$, since setting every field factor to the vacuum projector gives
    \begin{align}
        \label{eq:interacting_vacuum_coefficient}
        \ZZ_\epsilon(H,\bm0)
        =\left[\bra0e^{-\epsilon H_0}\prod_k T_k T_k^\dagger\ket0\right]^M
        =(1+O(\epsilon))^M\ne0.
    \end{align}
\end{proof}

\subsection{Proof of \cref{thm:Lee-Yang_theorem,thm:spectral_gap}}

Given \cref{thm:unit_disk_zero_free_interacting_fermions}, substituting $\bm u=e^{-\epsilon h}\bm1$ into $\ZZ_\epsilon(H,\bm u)$ implies that
\begin{align}
    \ZZ_\epsilon(h):=\ZZ_\epsilon(H,\bm u=e^{-\epsilon h}\bm1)
\end{align}
is zero-free when $\mathrm{Re}(h)>0$ for sufficiently large $M$. Since each field factor equals $\prod_i e^{-\epsilon h n_i}$ and each layer equals $e^{-\epsilon(H+h\sum_i n_i)}+O(\epsilon^2)$, we have $\ZZ_\epsilon(h)\to\ZZ(h)$ locally uniformly in $h$ as $M\to\infty$, with $\epsilon=\beta/M$. Since $\ZZ(h)>0$ for real $h>0$, Hurwitz's theorem implies that $\ZZ(h)\ne0$ throughout $\mathrm{Re}(h)>0$.

To bound the zeros to the strip $\mathrm{Re}(h) \in [-\lambda_{\max}(A), 0]$, apply the transformation $a \rightarrow a^{\dagger}$ on all the modes which changes 
\begin{align}
    n_i \quad&\rightarrow\quad 1-n_i,\\
    B \quad&\rightarrow\quad -B^* \\
    \sum A_{ij} a^{\dagger}_i a_{j} \quad&\rightarrow\quad \mathrm{Tr}(A) - \sum A_{ij}^{*} a_{i}^{\dagger} a_j, \\
    \text{and} \quad q^{\pm}_{v,w} \quad&\rightarrow\quad -q^{\mp}_{v^*,w^*}
\end{align} 
so that $-B^*$ is still an arbitrary antisymmetric matrix and $V$ remains of the form \cref{eq:interactions}. We can ignore the $\mathrm{Tr}(A)$ shift of the Hamiltonian $H_0$ since that does not affect the zeros of the partition function. The remaining matrix $-A^{*}$ is not PSD but the shifted matrix $\lambda_{\max}(A) \I - A^{*}$ is PSD. Absorbing $\lambda_{\max}(A)$ into the strength of the uniform field and applying the above argument, we get that the partition function $\ZZ(h)$ is also zero-free when $\mathrm{Re}(h) < -\lambda_{\max}(A)$. This completes the proof of \cref{thm:Lee-Yang_theorem}.

Starting from \cref{thm:unit_disk_zero_free_interacting_fermions}, the proof of \cref{thm:spectral_gap} essentially follows from the arguments in \cite{rayudu2026spectral} where the spin Hamiltonians are replaced by the fermionic Hamiltonians from \cref{def:lee-yang}, and the occupation-number basis plays the role of the spin computational basis. We remark that it is not clear whether the techniques from \cite{bravyi2026efficient} can also be extended to fermionic Lee-Yang Hamiltonians. This is because the Lee-Yang property of the corresponding generating polynomial used in the spin case does not hold in the fermionic case.

\section{\label{sec:discussion}Discussion}

In this work, we have established a Lee-Yang theorem for fermionic systems that includes many strongly interacting models of interest. As the theory of zero-freeness has important consequences for both physics and computer science, we discuss some implications of our results.

\paragraph{Efficient algorithms for simulating interacting fermions.}
Despite the immense interest in simulating correlated fermions, using either classical \cite{RevModPhys.68.13,PhysRevX.5.041041,zheng2017stripe} or quantum \cite{PhysRevA.92.062318,Kivlichan2020improvedfault,bauer2020quantum} computers, there is a notable paucity in the number of models for which a polynomial-time algorithm is guaranteed to compute the ground-state energy to small controllable error. On the side of classical algorithms, the known models are broadly limited to free fermions, gapped one-dimensional chains \cite{landau2015polynomial,arad2017rigorous}, and impurity models \cite{bravyi2017complexity,arunachalam2026quantum,jiang2026polynomial}. Even for the SPF Hamiltonians discussed in \cref{sec:sign_problem}, it is not known whether or not any of their corresponding Markov chains mix in polynomial time. Nonetheless, even if all SPF Hamiltonians were somehow classically easy, there is evidence that the non-SPF models, which are still Lee-Yang (such as the interacting Hofstadter model), may remain classically hard \cite{motruk2016density}.

The literature is similarly constrained when examining quantum algorithms. By a Jordan-Wigner mapping, the closely related results on quantum spin models obeying a Lee-Yang theorem \cite{rayudu2026spectral,bravyi2026efficient} also apply at least to one-dimensional fermion chains, even in the gapless case. Lindbladian-based algorithms for preparing ground states are a promising approach \cite{zhan2026rapid,ding2026simple}, but their rigorous analyses are currently restricted to the perturbative regime. The algorithmic consequence of our results therefore captures a vast swath of strongly interacting fermion models not previously known to be easy, either classically or quantumly.

\paragraph{\label{sec:scope-of-FLY}Extent of the fermionic Lee-Yang family.}

Although the Lee-Yang zero-free property may hold for broader families of fermionic Hamiltonians, the class studied here arises naturally from our proof strategy.
We first established the Lee-Yang property for free
fermions in \cref{sec:free-fermions}, and then extended it to the interacting Hamiltonians considered here via Asano contraction in \cref{sec:interacting}. 
The crucial step in this approach is to show that $q^2=0$ and $q\Gamma\preceq0$ which leads to \cref{eq:gammashrink}. 
This severely constrains the possible form of the elementary quadratic building blocks, essentially giving the form in \cref{eq:Q_operators}. 

The form of \cref{eq:Q_operators} can also be interpreted in a physically intuitive way. 
First, the mapping to spin Hamiltonians presented in \cref{sec:spin_hamiltonians} becomes more straightforward when a particle-hole symmetry is applied for a bipartite Hubbard model, because then the encoding Pauli matrices in \cref{eq:Pauli_encoding} will exactly correspond to the natural Pauli operators induced by the spins. 
The particle-hole transformation can be viewed as a specific type of {\it pseudospin} transformation, where any ``rotation'' that results in two orthonormal modes ($v$ and $w$) gives the same algebraic structure. 
In light of this view, the $q$-operators are simply raising and lowering operators $q^\pm_{v,w}=(-Z_{v,w}\pm iY_{v,w})/{2}$ with rotated axis. 
This choice of axis makes the associated field a pairing field, which is convenient for our formulation of the Lee-Yang property.

\paragraph{\label{sec:topological-order}Constraints on topological order.}

Here, we point out the implications for the existence of topological phases, in particular, two-dimensional gapped topological phases with ground-state degeneracy on a torus~\cite{wenniu1990,kitaev2003}.
According to \cref{thm:spectral_gap}, the gap of the system satisfies $\Delta_L(h)\geq ch$, where $c>0$ is a constant, for any fixed $h>0$.
Importantly, this lower bound does not depend on the system size $L$ or any coefficient strengths.
Suppose that a topological phase with ground-state degeneracy exists at $h=0$ and that this degeneracy is stable under sufficiently weak local perturbations.
For sufficiently small fixed $h>0$, let $\delta_L(h)=E_1(h)-E_0(h)$ denote the energy splitting between the two lowest-energy states.
This splitting must vanish as $L\to\infty$, yet the Lee-Yang spectral gap result implies
\begin{equation}
    \delta_L(h)=\Delta_L(h)\geq ch.
\end{equation}
Taking the thermodynamic limit at fixed $h>0$ therefore leads to a contradiction.
Thus, gapped topologically ordered phases with such stable ground-state degeneracy cannot occur in local systems for which the spectral gap theorem applies.

This result imposes a constraint on the half-filled repulsive Hubbard model on the honeycomb lattice, where numerical studies have been performed extensively and $h$ corresponds to a staggered Zeeman field.
Early quantum Monte Carlo calculations suggested the existence of a gapped phase between the semimetallic phase and the antiferromagnetically ordered phase~\cite{meng2010}.
However, subsequent studies of larger systems and improved finite-size analyses challenged this intermediate phase and strongly supported a direct and continuous transition between the semimetallic and antiferromagnetic phases~\cite{sorella2012,assaad2013,PhysRevX.6.011029}.
Under the stability assumption above, our result rules out interpreting the initially proposed gapped phase as a topologically ordered phase with ground-state degeneracy.
The direct and continuous semimetal-to-antiferromagnet transition supported by these subsequent studies is therefore consistent with our theorems.
We emphasize that our result does not exclude spin-gapped phases themselves; it does, however, exclude topologically ordered phases with the stable ground-state degeneracy as described above.

\section*{Acknowledgments}

We thank Ojas Parekh for suggesting a broader definition of the fermionic Lee-Yang class than the one presented in an earlier version of this manuscript.
C.R.\ acknowledges the funding support by UK Research and Innovation (UKRI) under the UK government's Horizon Europe funding guarantee EP/X032051/1.
T.M.\ was supported by JST FOREST Grant Number JPMJFR236N and by JSPS KAKENHI Grant Number JP26K00652.
J.T.\ was supported by JSPS KAKENHI Grant Numbers JP25K17310, JP25H01391, JP25H01388, JP25K24845, and JP25K24848. 
A.Z.\ was supported by the National Nuclear Security Administration’s Advanced Simulation and Computing program. 

This article has been authored by an employee of National Technology \& Engineering Solutions of Sandia, LLC under Contract No.\ DE-NA0003525 with the U.S. Department of Energy (DOE). The employee owns all right, title and interest in and to the article and is solely responsible for its contents. The United States Government retains and the publisher, by accepting the article for publication, acknowledges that the United States Government retains a non-exclusive, paid-up, irrevocable, world-wide license to publish or reproduce the published form of this article or allow others to do so, for United States Government purposes. The DOE will provide public access to these results of federally sponsored research in accordance with the DOE Public Access Plan \url{https://www.energy.gov/downloads/doe-public-access-plan}.

\subsection*{AI Usage Disclosure}

The authors used AI-based tools, such as Claude and OpenAI Codex, as aids for suggesting proof ideas.
All proofs and conclusions were independently verified by the authors.

\bibliographystyle{alphaurl}
\bibliography{refs}

\end{document}